\documentclass[fullpage,11pt]{article}

\usepackage{fullpage,amsthm}
\usepackage[T1]{fontenc}
\usepackage[utf8]{inputenc}

\usepackage{graphicx} 
\usepackage{array} 
\usepackage{amsmath, amssymb, amsfonts, color, verbatim}
\usepackage{epsfig,subfigure,multirow}
\usepackage[usenames,dvipsnames]{xcolor}
\usepackage{hyperref}
 \hypersetup{
  colorlinks=true,
  citecolor=ForestGreen,
}
\usepackage[a4paper,margin=1in]{geometry}
\usepackage{tcolorbox}

\usepackage{url}
\usepackage{xspace}
\usepackage[mathscr]{euscript}
\usepackage{algorithm2e, algorithmic}
\usepackage{cite}
\usepackage{enumerate}

\renewcommand{\qed}{\nobreak \ifvmode \relax \else
      \ifdim\lastskip<1.5em \hskip-\lastskip
      \hskip1.5em plus0em minus0.5em \fi \nobreak
      \vrule height0.75em width0.5em depth0.25em\fi}

\newcommand{\ssection}[1]{\medskip\noindent\textbf{#1}\xspace}
\renewenvironment{proof}[1][Proof]{\ssection{#1.}}{\hfill$\qed$}

\newtheorem{theorem}{Theorem}
\newtheorem{lemma}{Lemma}[section]

\newtheorem{claim}[lemma]{Claim}

\newtheorem{definition}{Definition}

\newtheorem*{claim*}{Claim}

\newcommand{\IZ}{\ensuremath{\mathbb{Z}}}

\newcommand{\floor}[1]{{\left\lfloor{#1}\right\rfloor}}

\newcommand{\set}[1]{\ensuremath{\{ #1 \}}}

\newcommand{\poly}{\mbox{\rm poly}}

\newcommand{\eps}{\varepsilon}

\newcommand{\REM}[1]{}

\newcommand{\union}{\ensuremath{\cup}}
\newcommand{\card}[1]{\left\vert{#1}\right\vert}

\newcommand{\rank}{\ensuremath{\textnormal{rank}}}
\newcommand{\lovasz}{Lov{\'a}sz}

\newenvironment{tbox}{\begin{tcolorbox}[
		enlarge top by=5pt,
		enlarge bottom by=5pt,
		 boxsep=0pt,
                  left=4pt,
                  right=4pt,
                  top=10pt,
                  arc=0pt,
                  boxrule=1pt,toprule=1pt,
                  colback=white
                  ]
	}
{\end{tcolorbox}}

\newcommand{\textbox}[2]{
{
\begin{tbox}
\textbf{#1} 
{#2}
\end{tbox}
}
}

\renewcommand{\bar}[1]{\overline{#1}}

\renewcommand{\vec}[1]{\mathbf{#1}}

\newcommand{\bv}{\vec{v}}

\newcommand{\bM}{\vec{M}}
\newcommand{\bA}{\vec{A}}

\newcommand{\CP}{\text{P}}

\newcommand{\mstalg}{\text{\sc MST-algorithm}}

\newcommand{\TT}{\ensuremath{T}}

\newcommand{\Qall}{\ensuremath{Q_{\forall}}}

\newcommand{\dout}[1]{d^{+}(#1)}
\newcommand{\din}[1]{d^{-}(#1)}
\newcommand{\stext}{\text{$s$-$t$ counterpart}}

\newcommand{\FG}{\mathcal{G}}
\renewcommand{\SS}{\mathcal{S}}

\DeclareMathOperator*{\Prob}{\ensuremath{\textnormal{Pr}}}
\renewcommand{\Pr}{\Prob}

\newcommand{\PR}[1]{\ensuremath{\Pr\left(#1\right)}\xspace}

\newcommand{\paren}[1]{\ensuremath{\left(#1\right)}\xspace}

\newcommand{\ourinfo}[1]{Department of Computer and Information Science, University of Pennsylvania. \newline Email: \texttt{#1}. 
Supported in part by National Science Foundation grants CCF-1116961, CCF-1552909, and IIS-1447470.}

\title{Dynamic Sketching for Graph Optimization Problems with Applications to Cut-Preserving Sketches} 
\author{Sepehr Assadi\thanks{\ourinfo{\{sassadi,sanjeev,yangli2\}@cis.upenn.edu}}\and 
Sanjeev Khanna\footnotemark[1] \and 
Yang Li\footnotemark[1] \and 
Val Tannen\thanks{Department of Computer and Information Science, University of Pennsylvania. Email: \texttt{val@cis.upenn.edu}. Supported in part by National Science Foundations grants IIS 1217798 and 1302212.}
}

\date{}

\begin{document}

\maketitle

\thispagestyle{empty}
\begin{abstract}
In this paper, we introduce a new model for sublinear algorithms called \emph{dynamic sketching}.  In this model, the underlying data is partitioned into a large \emph{static} part and a small \emph{dynamic} part and the goal is to 
compute a summary of the static part (i.e, a \emph{sketch}) such that given any \emph{update} for the dynamic part, one can combine it with the sketch to compute a given function. We say that a sketch is {\em compact} if its size is 
bounded by a polynomial function of the length of the dynamic data, (essentially) independent of the size of the static part.

A graph optimization problem $P$ in this model is defined as follows. The input is a graph $G(V,E)$ and a set $T \subseteq V$ of $k$ terminals; the edges between the terminals are the dynamic part and the other edges in $G$ are 
the static part. The goal is to summarize the graph $G$ into a compact sketch (of size $\poly(k)$) such that given any set $Q$ of edges between the terminals, one can answer the problem $P$ for the graph obtained by inserting all 
edges in $Q$ to $G$, using only the sketch.

We study the fundamental problem of computing a maximum matching and prove tight bounds on the sketch size. In particular, we show that there exists a (compact) dynamic sketch of size $O(k^2)$ for the matching problem and 
any such sketch has to be of size $\Omega(k^2)$. Our sketch for matchings can be further used to derive compact dynamic sketches for other fundamental graph problems involving cuts and connectivities.
Interestingly, our sketch for matchings can also be used to give an elementary construction of a  \emph{cut-preserving vertex sparsifier} with space $O(kC^2)$ for $k$-terminal graphs, which matches the best known upper bound; 
here $C$ is the total capacity of the edges incident on the terminals. Additionally, we give an improved lower bound (in terms of $C$) of $\Omega(C/\log{C})$ on size of cut-preserving vertex sparsifiers, and establish that 
progress on dynamic sketching  of the $s$-$t$ max-flow problem (either upper bound or lower bound) immediately leads to better bounds for size of cut-preserving vertex sparsifiers.

 \end{abstract}
\clearpage
\setcounter{page}{1}

\section{Introduction}\label{sec:intro-short}

Massive data sets are arising more and more frequently in many application domains. Traditional gold standards
of computational efficiency, namely, linear-time and linear-space, no longer seem sufficient for managing
and analyzing such massive data sets. As a result, a beautiful new area of sublinear algorithms has developed over the past two decades – these are algorithms whose resource requirements are substantially
smaller than the size of the input on which they operate. A rich theory of sublinear algorithms has emerged,
and has brought remarkable new insights into combinatorial structure of well-studied optimization problems (see, for instance, the surveys~\cite{mcgregor2014graph,muthu2005data,RubinfeldS11}, and references therein).

In recent years, graph optimization problems have received a lot of attention in the study of sublinear algorithms in various models, and the streaming model of computation is one of the most popular examples.
In the streaming model, an algorithm is presented with a stream of edge insertions and deletions and is required to give an answer to a pre-specified graph problem at the end of the stream. Unfortunately, for many
fundamental graph problems, no small space streaming algorithm is possible. For instance,~\cite{feigenbaum2004on} showed that determining whether or not there is a path from a specific vertex $s$ to a specified vertex $t$ in a 
directed graph requires $\Omega(n^2)$ space even for streams with only  edge insertions; here $n$ denotes the number of vertices in the input graph.
This immediately implies that computing the length of the $s$-$t$ shortest path, the
value of the minimum cut between $s$ and $t$, or the edge/vertex connectivity between $s$ and $t$ also requires $\Omega(n^2)$ space since the output of these problems is non-zero only when there is a path from $s$ to $t$. The 
same lower bound is also obtained for computing the size of the maximum matching~\cite{feigenbaum2004on}.
In fact, most of recent works for graph problem focus on \emph{approximation algorithms} developed under the \emph{semi-streaming} model
introduced in~\cite{feigenbaum2004on}, where an algorithm is allowed to output an approximate answer while using space linear in $n$. But is there hope left for \emph{exact} sublinear algorithms? More specifically, is there
 a non-trivial model where sublinear algorithms are achievable for outputting exact answers for fundamental graph problems like matchings, connectivities, cuts, etc.?

In this paper, we explore this direction by considering the case where the input graph only undergoes \emph{local changes}, and study how local changes
influence the solutions of several fundamental graph problems. The goal is to exploit the locality of these updates and compress the rest of the graph into a small-size \emph{sketch} that 
is able to answer \emph{queries} regarding a specific problem (e.g. the $s$-$t$ edge connectivity problem) for every possible local changes made to the graph. We introduce a model in this spirit and in the rest of this section, we 
formally define the model, discuss the connection to existing models, and summarize our results.

\subsection{The Dynamic Sketching Model}
We define the \emph{dynamic sketching model}, where algorithms are required to construct data structures (called sketches) that are composable
with \emph{local updates} to the underlying data.
Specifically, for graph problems in the dynamic sketching model, we consider the following setup. 
Given a graph optimization problem $P$, an input graph $G(V,E)$ on $n$ vertices with $k$ vertices
identified as \emph{terminals} $T = \set{q_1,\ldots,q_{k}}$, the goal of $k$-dynamic sketching for $P$ is to construct a sketch $\Gamma$  such that given any
possible subset of the edges between the terminals (a \emph{query}), we can solve the problem $P$ using only the information contained in the sketch $\Gamma$. Formally,

\begin{definition}\label{def:k-dynamic}
Given a graph-theoretic problem $P$, a $k$-dynamic sketching scheme for
$P$ is a pair of algorithms with the following properties.
\begin{enumerate}[(i)]
\item A \textbf{compression algorithm} that given any input graph $G(V,E)$ with a set $T$ of $k$ terminals,
outputs a data structure $\Gamma$ (i.e, a dynamic sketch).
\item An \textbf{extraction algorithm} that given any subset of the edges between
the terminals, i.e, a \textbf{query} $Q$, and the sketch $\Gamma$, outputs
the answer to the problem $P$ for the graph, denoted by $G^{Q}$, obtained by inserting all edges in $Q$ to $G$  (without further access to $G$).
\end{enumerate}
\end{definition}

We allow both compression and extraction algorithms to be randomized and err with some small probability. Furthermore, we say a sketching scheme is \emph{compact} if it constructs dynamic sketches
of size $\poly(k)$, where the size of a sketch is measured by the number of machine words of length $O(\log{n})$.

We should note right away that \emph{not} every graph problem admits a compact dynamic
sketch. For example, one can show that any dynamic sketch for the \emph{maximum clique}
problem or the \emph{minimum vertex cover} problem requires $\min\set{\Omega(n), 2^{\Omega(k)}}$
space (see Section~\ref{sec:np}; see also Section~\ref{sec:boolean} for a problem in $\CP$
where any dynamic sketch requires at least $\min\set{n,2^{k}}$ space).

In this paper, we focus on dynamic sketching for graph problems, but one should note that
the dynamic sketching model is not restricted to the graph problems; any problem/function
can be defined in this model as follows.  Given a function $f$, the input to a
\emph{dynamic sketching scheme} for $f$ is a data set whose elements are partitioned into
a \emph{static part} (data values for these elements remain fixed) and a \emph{dynamic
  part} (data values for these elements are allowed to change at run-time).  The goal is
to compress the data set into a sketch that contains enough information to recover the
value of the function $f$ (approximately or exactly) no matter how the dynamic part of the
data changes.

\subsection{Connection to Existing Models}
\subparagraph{Streaming.} Any single-pass streaming algorithm with space requirement $s$ can be used as a dynamic sketching scheme with a sketch of size $s$: run the streaming algorithm
 on graph $G(V,E)$ for the static data and store the state of the algorithm as the sketch; continue running the algorithm using the stored state when the dynamic data is presented.
 However, note that a streaming algorithm directly gives a \emph{compact} scheme only when the space requirement is logarithmic in $n$, which, as we just
discussed, is not the case for nearly all fundamental graph optimization problems. In the following, we use the $s$-$t$ shortest path problem as an example to elaborate the distinction between the two models.
Our results in Section~\ref{sec:matching}, illustrates a similar distinction for the case of the maximum matching problem.

As we already mentioned, outputting the length of the $s$-$t$ shortest path requires $\Omega(n^2)$ space in the streaming model.  We now give a simple dynamic sketching scheme for the $s$-$t$ shortest problem
with a sketch of size $O(k^2)$. The input to the $s$-$t$ shortest path problem in the $k$-dynamic sketching model is a weighted graph $G$, a set $T$ of terminals, and
two designated vertices $s$ and $t$.  Without loss of generality, we can assume $s$ and $t$ are terminals; otherwise we can add them to the set of terminals and record their edges to the other terminals in $O(k)$ space.
 The compression algorithm creates a graph $H$ with $V(H) =
\TT$, where for any pair of terminals $q_{i}$ and $q_{j}$, a directed
edge from $q_i$ to $q_j$ is added to $H$ with weight equal to the weight
of a shortest path from $q_i$ to $q_j$ in $G$. The size of $H$ is
$O(k^2)$. To obtain the
answer for each query $Q$, the extraction algorithm adds the edges
in $Q$ to $H$, building a small graph $H^{Q}$, and compute the
shortest path from $s$ to $t$ in $H^Q$.  It is easy to see that the weights of the shortest paths between $s$ and $t$ in
$H^{Q}$ and $G^{Q}$ are equal, thus $H$ can be used as a dynamic sketch for the $s$-$t$ shortest path problem.

\subparagraph{Linear sketches.} A very strong notion of sketching for handling arbitrary changes to the
original data is linear sketching, which corresponds to applying a randomized low-dimensional linear transformation to the input data.
This allows for compressing the data into a smaller space while (approximately) preserving some desired property of the input. Moreover,
composability of these sketches (as they are linear transformations) allow them to handle arbitrary changes to the input data.
Linear sketching technique has been successfully applied to various graph problems, mainly
involving cuts and connectivity~\cite{ahn2012graphS,ahn2012graphP,kapralov2014single}
(see also~\cite{mcgregor2014graph} for a survey of such results in
dynamic graph streams).  However, these results use space that is
prohibitively large for dynamic sketching (a linear dependence on $n$),
and typically only yield approximation answers.

\subparagraph{Kernelization.} Dynamic sketching shares some similarity to \emph{Kernelization} developed in parametrized
complexity~\cite{kratsch2012representative,kratsch2012compression,fortnow2008infeasibility} in the following two aspects. Firstly, the number of terminals $k$ in dynamic sketching may be viewed as a parameter. 
However, the main difference here is that for dynamic sketching, $k$ is a parameter of the \emph{model}, while for kernelization, the parameter is usually the size of the solution, which is the property of the input rather than the 
model. Secondly, since a kernel for an instance of a problem is defined to be an equivalent instance of the same problem with size bounded by a function
of a fixed parameter of the problem, both dynamic sketching and kernelization are in the spirit of compression. However, the techniques developed in kernelization do not directly carry over to dynamic sketching
for the following two reasons. Firstly, kernelization typically focuses only on static data and secondly, the space target in kernelization (which is different compare to dynamic sketching) is normally polynomial
in the parameter (usually the size of the solution to the problem) which could be $\Omega(n)$ in the dynamic sketching model. Finally,
it is worth mentioning that there are problems (e.g. \emph{minimum vertex cover}) that admit polynomial size kernels, while it can be shown  that the dynamic sketching for these problem
require sketches of size $2^{\Omega(k)}$ (see Section~\ref{sec:np}).

\subparagraph{Provisioning.} We should note that dynamic sketching shares some ancestry with \emph{provisioning},
a technique developed by~\cite{deutch2013caravan} for avoiding repeated expensive computations
in what-if analysis, where the input data is formed by $k$ known overlapping subsets
of some universe, and the goal is to compress these subsets so as to answer a specific database query
when only some of those subsets are presented at run-time. Note that a main distinction between the two model is that in provisioning the dynamic input is neither small nor local.

\subsection{Our Results}

\subparagraph{Maximum matching.} The main focus  of this paper is on the maximum matching problem~\cite{lovasz2009matching} in the dynamic sketching model, and its applications to various others problems. We give a dynamic
sketching scheme with a sketch of size $O(k^2)$, using a technique based on
an algebraic formulation of the matchings introduced by Tutte~\cite{tutte1947factorization}. At a high level, we store a sketch that computes the rank of the \emph{Tutte matrix} (see
Definition~\ref{def:tutte}) of the underlying graph. Since the queries only affect $O(k^2)$ entries of the Tutte matrix, we can compress this matrix using algebraic operations into a few small matrices of
dimensions $k \times k$. Storing the small matrices as the sketch and modifying the related entries when a query is presented allows us to compute the rank of the original Tutte matrix, and hence the maximum matching size.
Furthermore, we prove that our sketching scheme is optimal in terms of its space requirement (up to a logarithmic factor). In particular, we show that any dynamic sketching scheme for the matching problem
has to store a sketch of size $\Omega(k^2)$ bits. We emphasize that the lower bound is information-theoretic; it holds even if the compression and extraction algorithms are computationally unbounded.

\subparagraph{Cut-preserving sketches.} Interestingly, we discovered that our scheme for matchings can be used to design a \emph{cut-preserving sketch}, which is the information-theoretic version of a \emph{cut-preserving vertex
sparsifier}~\cite{hagerup1998characterizing,moitra2009approximation,leighton2010extensions}.
Given a capacitated graph $G$ (assume all capacities are integers) with a set $T$ of $k$ terminals, a cut-preserving vertex sparsifier (or a sparsifier for short) of $G$ is a graph $H$ with $T \subseteq V(H)$ ($V(H)$ denotes the set of vertices of $H$) such that for any bipartition $S$ and $T\setminus S$ of terminals, the value of the
minimum cut between $S$ and $T\setminus S$ in $G$ is preserved in $H$. A vertex sparsifier where the stored data is not restricted to be a graph is called a cut-preserving sketch.

In recent years, cut-preserving vertex sparsifiers have been extensively studied
(see, for example,~\cite{charikar2010vertex,englert2010vertex,chuzhoy2012vertex,andoni2014towards}).
For instance, exact sparsifiers with $2^{2^{k}}$
vertices are shown by~\cite{hagerup1998characterizing,khan2014mimicking}, and
sparsifiers with $O(C^3)$ vertices are shown by~\cite{kratsch2012representative}, where $C$ is the
total capacity of the edges incident on the terminals. Additionally, the size of any exact sparsifier is shown to be
$2^{\Omega(k)}$~\cite{krauthgamer2013mimicking,khan2014mimicking}.
Cut-preserving sketches are also studied in the literature~\cite{andoni2014towards,krauthgamer2013mimicking,kratsch2012compression},
where the best construction is known to be of size $O(kC^{2})$
by~\cite{kratsch2012compression}. Moreover, the $2^{\Omega(k)}$ lower bound of~\cite{krauthgamer2013mimicking} is also shown to hold for the cut-preserving sketches.

We show that our dynamic sketching scheme for matchings can be used to obtain an elementary construction of a cut-preserving sketch of size $O(kC^2)$ that matches the
best known upper bound of~\cite{kratsch2012compression}. \cite{kratsch2012compression} showed that given a graph $G$ and a set of $k$ terminals $T$, a single \emph{gammoid} 
can be used to produce a matroid that encodes all \emph{terminal vertex cuts}. The authors then use the result of~\cite{Marx09} to show how to obtain a matrix representation
of this gammoid with $O(k^2)$ entries of $O(k)$ bits each (see Corollary 3.2 of~\cite{kratsch2012compression}). Using standard techniques, one can use this sketch for 
vertex cuts to obtain an sketch for edge cuts (i.e, a cut-preserving sketch) that requires $O(kC^2)$ space. Our construction, on the other hand, uses 
the connection between matchings and the Tutte matrix followed by a simple reduction from cut-preserving sketches to the maximum matching problem. 
We believe that the simplicity of this construction and its connection to dynamic sketches for the matching problem is
of independent interest and gives further insights into the structure of cut-preserving sketches.  
Moreover, we prove an improved lower bound (in terms of $C$) of $\Omega(C/\log{C})$ bits on the size of
any cut-preserving sketch; prior to our work, the best lower bound in terms of $C$ is $\Omega(C^\eps)$ for some
small constant $\eps > 0$ obtained by~\cite{krauthgamer2013mimicking}.

\subparagraph{$s$-$t$ edge-connectivity and $s$-$t$ maximum flow.} As it turns out, any cut preserving sketch can be (almost directly) used to obtain a dynamic sketching scheme for the $s$-$t$ edge-connectivity problem. However,
using our lower bound for cut-preserving sketches, the resulting sketch size for edge-connectivity would be $\Omega(C/\log{C})$, where $C$ could be as large as $n$ (hence the sketch is not compact). To obtain a compact sketch
for edge-connectivity, we further design a dynamic sketching scheme which directly uses our dynamic sketching scheme for matchings, and obtain compact sketches of size $O(k^4)$. We further establish that cut-preserving 
sketches are, in fact, more related to the $s$-$t$ maximum flow problem, in the sense that progress on either upper bound or lower bound on size of dynamic sketches for the $s$-$t$ maximum flow problem immediately leads to 
better bounds for size of cut-preserving sketches.

\subparagraph{Minimum spanning tree.} Finally, we present an $O(k)$-size dynamic sketch for the minimum spanning tree (MST) problem. 
Our idea for creating a compact dynamic sketch for MST is as follows. First of all, it is easy to see that if we add an edge to a graph, an MST of the resulting graph can be created by adding the edge to an MST of the original 
graph. Hence, it is sufficient to store an MST $H$ of the original graph as a sketch. But this sketch is of size $\Omega(n)$. We show that $H$ can be compressed into a tree 
$H'$ such that all leaf nodes are terminals and there are at most $O(k)$ internal nodes in this tree; moreover, for any query $Q$, the weights of the MSTs in $G^Q$ and $H'^Q$ are equal. Hence, $H'$ can be stored as a dynamic 
sketch.

\subparagraph{Organization.} The rest of the paper is organized as follows.
We first introduce our dynamic sketching scheme for the maximum matching problem in Section~\ref{sec:matching} and prove its optimality in terms of the sketch size. 
Then, in Section~\ref{sec:sparsifier-upper}, we show how to use our
sketching scheme for matching problem to construct a cut-preserving sketch. Next, we provide our improved lower bound on size of cut-preserving sketches in Section~\ref{sec:sparsifier-lower}.
We further establish the connection between cut-preserving sketches and $s$-$t$ edge-connectivity and introduce a compact dynamic
sketch for edge connectivity in Section~\ref{sec:st-conn-matching}. We present the equivalence between $s$-$t$ maximum flow and cut-preserving sketches in Section~\ref{sec:max-flow}. 
Our dynamic sketch for the minimum spanning tree problem is provided in Section~\ref{sec:mst}. We also provide two simple lower bounds in Section~\ref{sec:limit} to demonstrate the limitations of 
the dynamic sketching model. Finally, we conclude the paper with some future directions in Section~\ref{sec:conclusion}. 

\subparagraph{Notation.} We denote by $[n]$ the set $\set{1,2,\ldots,n}$. The bold-face upper-case letters represent  matrices.
A matrix with a `tilde' on top (e.g. $\widetilde{\bM}$) always denotes a symbolic matrix, i.e, a matrix containing formal variables.
For any prime $p$, $\IZ_p$ denotes the field of integers modulo $p$.

For any undirected graph $G$, we use $\nu(G)$ to denote the size of a maximum matching in $G$.
For any directed graph $G(V,E)$, an edge $e = (u,v)$ is directed from $u$ to $v$, where we say $u$ is the
\emph{tail} and $v$ is the \emph{head} of $e$. For any vertex $v \in V$, $\dout{v}$ (resp. $\din{v}$) denotes the number of outgoing
(resp. incoming) edges of $v$. For a capacitated graph, $c^{+}(v)$ (resp. $c^{-}(v)$) denotes the total capacity of the outgoing (resp. incoming) edges of $v$.
We assume all the capacities are integers and can be stored in a single machine word of size $O(\log{n})$.

\newcommand{\symbolic}{\widetilde}
\newcommand{\tvM}{\ensuremath{\symbolic{\bM}}}
\newcommand{\vA}{\symbolic{\mathbf{A}}}
\newcommand{\vB}{\symbolic{\mathbf{B}}}
\newcommand{\vC}{\symbolic{\mathbf{C}}}
\newcommand{\vD}{\symbolic{\mathbf{D}}}
\newcommand{\MA}{{\mathbf{A}}}
\newcommand{\MB}{{\mathbf{B}}}
\newcommand{\MC}{{\mathbf{C}}}
\newcommand{\MD}{{\mathbf{D}}}
\newcommand{\MX}{{\mathbf{X}}}
\newcommand{\MY}{{\mathbf{Y}}}
\newcommand{\MI}{{\mathbf{I}}}
\newcommand{\MO}{{\mathbf{0}}}

\newcommand{\tM}{\ensuremath{\bM}}
\newcommand{\diag}{\ensuremath{\textnormal{diag}}}
\newcommand{\scSize}{\ensuremath{\eta}}
\newcommand{\apply}[2]{{#1}_{\mid #2}}

\newcommand{\membership}{\textnormal{\textsc{Membership}}\xspace}

\section{The Maximum Matching Problem}\label{sec:matching}

In this section, we provide our results for the maximum matching problem. In particular, we show that,  

\begin{theorem}\label{thm:matching}
  For any $0 < \delta < 1$, there exists a randomized $k$-dynamic sketching scheme for the
  maximum matching problem with a sketch of size $O(k^2 \log{(1/\delta)})$, which answers any query correctly with probability at least $1 - \delta$.
 \end{theorem}

Furthermore, we prove that the sketch size obtained in Theorem~\ref{thm:matching} is tight (up to an $O(\log{n})$ factor). Formally, 

\begin{theorem}\label{thm:matching-lower}
	For any $k \geq 2$, any $k$-dynamic sketching scheme for the maximum matching problem
	that answers any query correctly with probability at least $2/3$, requires a dynamic sketch of size $\Omega(k^2)$ bits.
\end{theorem}

Our sketching scheme for the proof of Theorem~\ref{thm:matching} relies on an algebraic formulation
for the matching problem due to Tutte~\cite{tutte1947factorization}.  In the remainder of this section, we present this algebraic formulation, state our sketching scheme for matchings and proves its correctness and then
present our lower bound result.

\ssection{Algebraic formulation for the matching problem.}
The following matrix was first introduced by Tutte~\cite{tutte1947factorization}.
\begin{definition}[Tutte matrix~\cite{tutte1947factorization}]\label{def:tutte}
  Suppose $G(V,E)$ is an undirected graph. The \emph{Tutte matrix} of
  $G$ is the following \emph{symbolic} matrix $\tvM$ of dimension $n \times n$.
  \[
  \tvM_{i,j} =
  \begin{cases}
  x_{i,j} &\mbox{if $(i,j) \in E$ and $i < j$}\\
  -x_{j,i} &\mbox{if $(i,j) \in E$ and $i > j$} \\
  0 &\mbox{otherwise}
  \end{cases}
\]
where the $x_{i,j}$ are distinct formal variables.
\end{definition}

Tutte~\cite{tutte1947factorization} discovered an important connection between perfect matchings of a graph $G$ and the above matrix: the determinant of the Tutte matrix (which is a polynomial of degree $n$ in the $\set{x_{i,j}}$ variables) is identical to zero if and only if $G$ does not have a perfect matching. This relation was later generalized to maximum matchings~\cite{lovasz2009matching}: the rank of the Tutte matrix is  twice the size of the maximum matching. Based on this relation and the Schwartz-Zippel Lemma~\cite{schwartz1980fast}, \lovasz~\cite{lovasz1979determinants} proposed a simple algorithm for computing the size of a maximum matching in a general graph summarized in the following Lemma (see also~\cite{rabin1989maximum} for more details on performing the computations over a finite field).

\begin{lemma}[\!\!\cite{lovasz1979determinants,rabin1989maximum}]\label{lem:rank-lovasz}
   Let $G$ be an undirected graph with $n$ vertices and the maximum matching size of $\nu(G)$. For any prime $p > n$, let $\IZ_p$ be the field of integers modulo $p$. Suppose $\tvM$ is the
   Tutte matrix of $G$ and $\tM$ is the matrix obtained by evaluating each variable in $\tvM$ by a number chosen independently and uniformly at random from $\IZ_p$; then:
   \[
	 \PR{\rank(\tM) = 2\nu(G)} \geq 1-\frac{n}{p}
   \]
   Note that the computation of $\rank(\tM)$ is also done over the field $\IZ_p$.
\end{lemma}

\subsection{An $O(k^2)$ size upper bound}

In this section, we provide our $k$-dynamic sketching scheme for the maximum matching problem and prove Theorem~\ref{thm:matching}. 

\ssection{Notation.} Suppose the input is an undirected graph $G(V,E)$ with a set $T = \set{q_1,\ldots,q_k}$ of $k$ terminals. 
Let $p$ be any prime of magnitude $\Theta(n/\delta)$; we perform the algebraic computations in the field $\IZ_p$. 
Let $\tvM$ be the Tutte matrix of the graph obtained by adding all edges between the terminals to $G$, where the first $k$ rows
and $k$ columns correspond to the vertices in $T$. We decompose $\tvM$ into four sub matrices $\vA,\vB,\vC$, and $\vD$ as follows:
\[
	\tvM = \begin{bmatrix}
	\vA_{k \times k} & \vB_{k \times (n-k)} \\
	\vC_{(n-k) \times k} & \vD_{(n-k) \times (n-k)}
	\end{bmatrix}
\]

\ssection{Compression algorithm:} The compression algorithm consists of $4$ steps. Each of them performs a simple algebraic manipulation on the Tutte matrix $\tvM$.

\ssection{Step~1.} For each non-zero entry of $\tvM$ that corresponds to an edge in $G$ (i.e., not between the terminals), assign an integer chosen uniformly at random from $\IZ_p$. Denote the resulting matrix by,
\[
\tvM_1 = \begin{bmatrix}
	\vA_{k \times k} & \MB_{k \times (n-k)} \\
	\MC_{(n-k) \times k} & \MD_{(n-k) \times (n-k)}
	\end{bmatrix}
\]
Note that except for $\vA$, all sub-matrices in $\tvM_1$ are \emph{no longer} symbolic.

\ssection{Step~2.} Let $r = \rank(\MD)$. Use \emph{elementary row and column operations} to change $\MD$ into a \emph{diagonal} matrix $\diag(1,\ldots,1,0,\ldots,0)$ with only $r$ non-zero entries.
Note that after this process, matrices $\MB$ and $\MC$ would also change, but the symbolic matrix $\vA$ remains unchanged. We denote the matrix $\tvM_1$ after this process by,
\[
\tvM_2 = \begin{bmatrix}
	\vA_{k \times k} & \MX_{k \times r} & \MB'_{k \times (n-k-r)} \\
	\MY_{r \times k} & \MI_{r \times r} & \MO_{r \times (n-k-r)} \\
	\MC'_{(n-k-r) \times k} & \MO_{(n-k-r) \times r}  & \MO_{(n-k-r)\times(n-k-r)}
	\end{bmatrix}
\]

\ssection{Step~3.} Use the sub-matrix $\MI_{r \times r}$  in $\tvM_2$ to zero out the matrix $\MX$ by elementary \emph{row} operations. Similarly, zero out $\MY$ by elementary \emph{column} operations. Note that after this process, the matrix $\vA$ would be added by a linear combination of the rows in $\MY$, denoted by $\MA'$. Denote the resulting matrix by,
 \[
\tvM_3 = \begin{bmatrix}
	\vA_{k \times k} + \MA'_{k \times k}& \MO_{k \times r} & \MB'_{k \times (n-k-r)} \\
	\MO_{r \times k} & \MI_{r \times r} & \MO \\
	\MC'_{(n-k-r) \times k} & \MO & \MO
	\end{bmatrix}
\]

\ssection{Step~4.} Consider the matrix $\MB'$ in $\tvM_3$; pick a maximal set of \emph{linearly independent} columns from $\MB'$ (if less than $k$ columns are picked, arbitrarily pick from the remaining columns until having picked $k$ columns), denoted by $\MB''_{k \times k}$.
Do the same for the matrix $\MC'$ (but using linearly independent rows) and create $\MC''_{k \times k}$. Finally, pick $k^2$ numbers from $\IZ_p$, independently and uniformly at random and form a matrix of dimension $k \times k$, denoted by $\hat{\bA}$.
Store the value $r$ (i.e., the rank of $\MD$), the matrix $\hat{\bA}$, and three $k \times k$ matrices $\MA'$, $\MB''$ and $\MC''$ as the sketch.

\ssection{Extraction algorithm:} Given a query $Q$, create the matrix $\hat{\bA}_Q$ from $\hat{\bA}$ by zeroing out every entry that corresponds to an edge \emph{not} in $Q$. Evaluate $\vA$ by $\hat{\bA}_Q$ and obtain a (non-symbolic) matrix $\MA$. Construct a matrix $\hat{\bM}$ as follows,
\[
	\hat{\bM} = \begin{bmatrix}
	\MA_{k \times k} + \MA'_{k \times k}& \MB''_{k \times k} \\
	\MC''_{k \times k} &   \MO_{k \times k}
	\end{bmatrix}
\]
Return $\paren{\rank(\hat{\bM})+r}/2$ as the maximum matching size.

We now prove the correctness of this scheme and show that it satisfies the bound given in Theorem~\ref{thm:matching} and hence prove this theorem.

\begin{proof}[Proof of Theorem~\ref{thm:matching}]
   Since the prime $p$ is of magnitude $\Theta(n/\delta)$, any number in $\IZ_p$ requires $O(\log(n/\delta)) = O(\log{n} + \log(1/\delta))$ bits to store, which
    is at most $O(\log(1/\delta))$ machine words. The compression algorithm stores a number $r$, which needs $O(\log{n})$ bits, four matrices of dimension $k \times k$,
    where each entry is a number in $\IZ_p$ and requires $O(\log(1/\delta))$ space. Therefore, the total sketch size is $O(k^2 \log(1/\delta))$. We now prove the correctness.

    We need to show that for each query $Q$, the extraction algorithm correctly outputs the matching size with probability at least $1 - \delta$.
    By Lemma~\ref{lem:rank-lovasz},
    \[
    \PR{\rank(\tM) = 2\nu(G^Q)} \geq 1 - \frac{n}{p} \ge 1 - \delta
    \]
     Here $\tM$ is the (randomly evaluated) Tutte matrix of the graph obtained by applying the query $Q$ to $G$, i.e., $G^Q$.
     Since the extraction algorithm outputs $\paren{\rank(\hat{\bM})+r}/2$ as the matching size, it suffices for us to
    show that $\rank(\hat{\bM}) + r = \rank(\tM)$.

    More specifically, the extraction algorithm evaluates $\vA$ by assigning a (pre-selected) random number to each entry that corresponds to an edge in $Q$, i.e, the matrix $\hat{\bA}_Q$. For the sake of analysis, assume this is
    done before the compression algorithm is executed. Then, at the first step of the compression algorithm, all entries of the matrices $\MB, \MC, \MD$ are  randomly and independently evaluated.
    Combined with evaluating $\vA$ by $\hat{\bA}_Q$, the resulting matrix (denoted by $\bM_1$) is obtained from randomly and independently evaluating
    every non-zero entry of the Tutte matrix of the graph $G^Q$. In other words, it suffices to show that $\rank(\bM_1) =  \rank(\hat{\bM}) + r$.

    Since step~2 and step~3 only perform elementary row/column operations on the matrix, the rank does not change. For the matrix $\tvM_3$ obtained after step~3, denote by $\bM_3$ the matrix
    after evaluating the $\vA$ part in  $\tvM_3$. $\bM_3$ is non-symbolic and it suffices to prove that $\rank(\bM_3) = \rank(\hat{\bM}) + r$. Note that after reordering rows and columns of $\bM_3$, $\bM_3$ can be rewritten as
 \[
    \begin{bmatrix}
	\bA_{k \times k} + \MA'_{k \times k} & \MB'_{k \times (n-k-r)} & \MO\\
	\MC'_{(n-k-r) \times k}                       & \MO & \MO \\
     \MO                                 & \MO & \MI_{r \times r} \\
	\end{bmatrix}
\]
       Therefore, the rank of $\bM_3$ is equal to $r$ plus the rank of the following sub-matrix of $\bM_3$.
 \[
    \bM_4 = \begin{bmatrix}
	\bA_{k \times k} + \MA'_{k \times k} & \MB'_{k \times (n-k-r)} \\
	\MC'_{(n-k-r) \times k}                       & \MO\\
	\end{bmatrix}
\]
    We now show that $\bM_4$ has the same rank as $\hat{\bM}$. Since the matrix $\MC''$ (in step~4 of the compression algorithm) contains a maximal set of linearly independent rows of $\MC'$, each remaining row of $\MC'$ is a linear combination of the rows in $\MC''$. Therefore, all remaining rows in $\MC'$ can be zero-out using elementary row operations. Hence, the rank of $\bM_4$ is equal to the rank of the following matrix
 \[
    \bM_5 = \begin{bmatrix}
	\bA_{k \times k} + \MA'_{k \times k} & \MB'_{k \times (n - k - r)} \\
	\MC''_{k \times k}                       & \MO \\
    \MO_{(n-2k-r) \times k}                       & \MO \\
	\end{bmatrix}
\]
    Similarly, using elementary column operations, the sub-matrix $\MB'$ in $\bM_5$ can be made into $[\MB''~~\MO_{k \times (n - 2k -r)}]$ without changing the rank, and the resulting matrix has the same rank as $\hat{\bM}$. 
    
    \end{proof}

\subsection{An $\Omega(k^2)$ size lower bound}

In this section, we prove an $\Omega(k^2)$ bits lower bound on the sketch size of any $k$-dynamic sketching scheme for the matching problem, which implies that our space upper bound
in Theorem~\ref{thm:matching} is tight (up to a logarithmic factor). 
We establish this lower bound by reducing from the \membership problem studied in communication complexity defined as follows.

\textbox{The \membership Problem}
{
\\
\textbf{Input}: Alice is given a set $S \subseteq [N]$ and Bob is given an element $e^* \in [N]$. \\
\textbf{Goal}: Alice has to send a message to Bob such that Bob can determine whether $e^* \in S$ or not.
}

It is well-known that in order for Bob to succeed with probability at least $2/3$, Alice has to send a message of size $\Omega(N)$ bits~\cite{Ablayev96}, where the probability is taken 
over the random coin tosses of Alice and Bob.

\ssection{Reduction.} For simplicity, assume $N$ is a perfect square. Given any $S \subseteq [N]$, Alice constructs a graph $G(V,E)$ with a set $T$ of $k$ terminals as follows:
\begin{itemize}
	\item The vertex set $V = \set{u,w} \union V_1 \union V_2 \union V_3 \union V_4$, where $\card{V_i} = \sqrt{N}$ for any $i \leq 4$ and $T = \set{u,w} \union V_1 \union V_4$. 
	We will use $v^{(i)}_j$ to denote the $j$-th vertex in $V_i$.
	\item For any $i \in [\sqrt{N}]$, $v^{(1)}_i$ (resp. $v^{(3)}_i$) is connected to $v^{(2)}_i$ (resp. $v^{(4)}_i$); i.e, there is perfect matching between $V_1$ and $V_2$ (resp. $V_3$ and $V_4$). 
	\item Fix a bijection $\sigma: [N] \mapsto [\sqrt{N}] \times [\sqrt{N}]$; for any element $e \in S$ with $\sigma(e) = (i,j)$,  $v^{(2)}_i$ is connected to $v^{(3)}_j$.
\end{itemize}
Note that in this construction, $n = 4\sqrt{N} + 2$ and $k = 2\sqrt{N} + 2$, and initially there is no edge between the terminals. 

Alice constructs this graph, run the compression algorithm of the dynamic sketching scheme on it and sends the sketch to Bob. 
Let $Q$ be the query in which, for $\sigma(e^*) = (i,j)$, $u$ is connected to $v^{(1)}_i $ and $v^{(4)}_j$ is connected to $w$. Bob queries the sketch with $Q$, finds the maximum 
matching size in $G^{Q}$, and returns $e^* \in [S]$ iff the maximum matching size is $2\sqrt{N} + 1$ in $G^{Q}$.

\begin{lemma}\label{lem:match-lower}
	 $\nu(G^{Q}) = 2\sqrt{N} + 1$ if and only if $e^* \in S$.
\end{lemma}
\begin{proof}
	Let $\sigma(e^*) = (i,j)$. Consider the matching $M$ where  $u$ (resp. $w$) is matched with $v^{(1)}_i$ (resp. $v^{(4)}_j$), and $v^{(1)}_{i'}$ (resp. $v^{(3)}_{j'}$)
	is matched with $v^{(2)}_{i'}$ (resp. $v^{(4)}_{j'}$) for any $i' \neq i$ (resp. $j' \neq j$). Note that $\card{M} = 2\sqrt{N}$ and both $v^{(2)}_i$ and $v^{(3)}_j$ are unmatched in $M$.
	
	First, suppose $e^* \in S$; in this case, the edge $e = (v^{(2)}_i,v^{(3)}_j)$ is in $G$ and can be added to $M$. Matching $M \union \set{e}$ is a perfect matching; hence, in this case, $\nu(G^Q) = 2\sqrt{N} + 1$.
	
	Now, suppose $e^* \notin S$; we claim $M$ is a maximum matching in this case. Indeed, the only two vertices that are unmatched in $G^Q$ are $v^{(2)}_i$ and $v^{(3)}_j$. Hence, any augmenting
	path of $M$ must have $v^{(2)}_i$ and $v^{(3)}_j$ as endpoints. However, there is no edge between $v^{(2)}_i$ and $v^{(3)}_j$, and moreover any alternative 
	path that starts from $v^{(2)}_i$ (resp. $v^{(3)}_j$), has length at most $2$. Consequently, there is
	no augmenting path for $M$ and hence $M$ is a maximum matching, implying that $\nu(G^Q) = 2\sqrt{N}$.
	
\end{proof}

Theorem~\ref{thm:matching-lower} now follows from Lemma~\ref{lem:match-lower}, along with the lower bound of $\Omega(N) = \Omega(k^2)$ on the 
communication complexity of the \membership problem. 

\newcommand{\eminus}{\ensuremath{e^{-}}\xspace}
\newcommand{\eplus}{\ensuremath{e^{+}}\xspace}
\newcommand{\qtoe}{\ensuremath{q^{\rightarrow e}}\xspace}
\newcommand{\qtoep}[2]{\ensuremath{{#1}^{\rightarrow {#2}}}\xspace}
\newcommand{\qfrome}{\ensuremath{q^{\leftarrow e}}\xspace}
\newcommand{\qfromep}[2]{\ensuremath{{#1}^{\leftarrow {#2}}}\xspace}
\newcommand{\qab}{\ensuremath{Q_{A,B}}\xspace}
\newcommand{\mstar}{\ensuremath{M^*}\xspace}

\newcommand{\op}{\textnormal{\textbf{op}}}
\newcommand{\tc}{\textnormal{TC}}

\section{Cut-Preserving Sketches}\label{sec:sparsifier}

We establish a connection between $k$-dynamic sketching schemes for the maximum matching problem and cut-preserving sketches. In particular, we use our dynamic sketching scheme for the matching
problem in Section~\ref{sec:matching} to design an exact cut-preserving sketch (i.e, an information-theoretic vertex sparsifier) with size $O(kC^2)$, where $C$ is the total capacity of the edges incident on the terminals. This
matches the best known upper bound on the space requirement of cut-preserving sketches. We further provide an improved lower bound of $\Omega(C/\log{C})$ on the size of any cut-preserving sketch.
Throughout this section, we will use the term \emph{bipartition cut} to refer to a cut between a bipartition of the terminals and the term \emph{terminal cut} to refer to a cut which separates
two arbitrary disjoint subsets of terminals (not necessarily a bipartition). With a slight abuse of notation, we refer to the value of the minimum cut for a bipartition/terminal cut as the value of the bipartition/terminal cut directly.

Before we present our results, we make a general remark about the property of all cut-preserving sparsifiers (and sketches) that is also used crucially in our lower bound proof.
In~\cite{krauthgamer2013mimicking}, a \emph{generalized} sparsifier is defined as a sparsifier that preserves the minimum cut between
all disjoint subsets of terminals, i.e, terminal cuts and not only bipartition cuts. The authors then point out that their upper bound results, as well as the previous constructions
of cut sparsifiers in~\cite{hagerup1998characterizing}, also satisfy this general definition. The following simple claim gives an explanation why all
known cut sparsifiers satisfy this general definition.

\begin{claim}\label{clm:general-cs}
	Suppose $H$ is a cut sparsifier of the graph $G(V,E)$ with terminals $T$ that preserves the value of all bipartition cuts. Then, $H$
	also preserves the value of all terminal cuts.
\end{claim}
\begin{proof}
For any two disjoint subsets of terminals $A,B \subseteq T$, any cut separating $A$ and $B$ in $G$ must form a bipartition $(S,\bar{S})$ of the terminals, and
since $H$ preserves the value of all minimum cuts like $(S,\bar{S})$, the $(A,B)$ minimum cut value in $H$ is also equal to the minimum cut value in $G$.
In the case that $H$ is a cut-preserving sketch, the $(A,B)$ minimum cut
can be answered by querying $H$ with all bipartition cuts that separate $(A,B)$, and outputting the smallest value.

\end{proof}

\subsection{An $O(k C^2)$ Size Cut-Preserving Sketch}\label{sec:sparsifier-upper}

In this section, we construct a cut-preserving sketch for any digraph $G$ and a set of terminals $T$. We achieve this by constructing an
instance $G'$ of the maximum matching problem in the dynamic sketching model and
show that the value of any terminal cut $(A,B)$ in $G$ can be computed using a carefully designed query for the maximum matching size in $G'$.
Our reduction is based on a classical result relating edge connectivity and bipartite matching due to~\cite{hoffman1960some} (see also~\cite{schrijver2003combinatorial}, Section $16.7$).

\begin{theorem}\label{thm:vertex-sparsifier-upper}
  For any directed graph $G$ with a set of $k$ terminals, there is an exact cut-preserving sketch that uses space $O(k C^2)$, where $C$ is the total capacity of the edges incident on the terminals.
\end{theorem}
Without loss of generality, we will replace each edge in $G$ with capacity $c_e$ with $c_e$ parallel edges and still denote the new graph with $G$. Consider the following cut-preserving sketch.
\textbox{A cut-preserving sketch}{
\begin{enumerate}
  \item[] \textbf{Input:} A graph $G$ with $m$ edges and a set $T$ of terminals.
  \item[$\bullet$] \textbf{Compression:} Construct a bipartite graph $G'(L,R,E')$ with terminals $T'$ as follows and create a dynamic sketch for the maximum matching problem for $G'$ and $T'$.
    \begin{enumerate}
        \item For each edge $e$ in $G$, create two vertices \eminus (in $L$) and \eplus (in $R$).
        \item For any terminal $q$ in $T$ and any outgoing (resp. incoming) edge $e$ of $q$, create a vertex \qtoe in $R$ (resp. \qfrome in $L$). \qtoe (resp. \qfrome), along with \eminus (resp. \eplus), belongs to the set $T'$ of terminals in $G'$.
        \item For each edge $e$ in $G$, there is an edge between vertices \eminus and \eplus in $G'$.
        \item For any two edges $e_1$ and $e_2$ in $G$ where the tail of $e_1$ is the head of $e_2$, there is an edge between the vertices $\eplus_1$ and $\eminus_2$ in $G'$.
    \end{enumerate}
    \item[$\bullet$] \textbf{Extraction:} Given any two disjoint subsets $A, B \subseteq T$, let \qab be the query where for any terminal $q$ in $A$ (resp. in $B$) and any outgoing (resp. incoming) edge $e$ of $q$, an edge between
    the vertices \qtoe and \eminus (resp. \qfrome and \eplus) is inserted in $G'$.

    Return $\nu(G'^{\qab}) - m$, where $\nu(G'^{\qab})$ is the maximum matching size in $G'^{\qab}$.
\end{enumerate}
}

The total number of terminals in $G'$ is $2C$, and the total number of different $A$-$B$ pairs (i.e., the total number of different possible queries) is at most $3^k$.
To ensure that every query is answered correctly, by Theorem~\ref{thm:matching}, the sketch size is  $O(k C^2)$. We now prove the correctness.

\begin{proof}
    For any $A,B \subseteq T$, denote the value of the minimum $A$-$B$ cut (which is equal to the edge-connectivity from $A$ to $B$), by $c(A,B)$. Recall that for a graph $G$, $\nu(G)$ denotes the maximum matching size
    in $G$.  We prove that $\nu(G'^{\qab}) - m = c(A,B)$. Let $M$ be the matching in $G^{\qab}$ where for each edge $e$ in $G$, \eminus is matched with \eplus; hence $\card{M} = m$.

    We first show that if $c(A,B) = l$, then $M$ can be augmented by $l$ vertex disjoint paths and hence $\nu(G'^{\qab}) \geq m+l$. There are $l$ edge-disjoint path $P_1, P_2, \ldots, P_l$ from $A$ to $B$ in $G$. For each
    path $P_i = (e_1, e_2, \ldots, e_j)$, where $e_1$ starts with a terminal $q_a \in A$ and $e_j$ ends with a terminal $q_b \in B$,  create a
    path $P'_i = (\qtoep{q_a}{e_1}, \eminus_1, \eplus_1, \eminus_2, \ldots, \eplus_j, \qfromep{q_b}{e_j})$ in $G'^{\qab}$. It is straightforward to
    verify that the $P'_i$ paths are valid \emph{vertex-disjoint} paths in $G'^{\qab}$ and moreover, form disjoint augmenting paths of the matching $M$. 

    We now show that if the maximum matching $\mstar$ in $G'$ is of size $m + l$, then $c(A,B) \ge l$. The symmetric difference between $M$ and $\mstar$ forms a graph with $l$ augmenting paths of the matching $M$.  Each
    augmenting path must start and end with a vertex of the form \qtoe or \qfrome since they are the only vertices that are unmatched in $M$. Since every \qtoe is in $R$ and every $\qfrome$ is in $L$, each augmenting path must
    start with a \qtoe vertex and ends with a \qfrome vertex. Using the reversed transformation as in the previous case, the $l$ augmenting paths can be converted into $l$ edge disjoint paths from $A$ to $B$ in $G$.
\end{proof}

\subsection{An $\tilde{\Omega}(C)$ Size Lower Bound} \label{sec:sparsifier-lower}

In this section, we provide a lower bound on the size of any cut-preserving sketch.

\begin{theorem}\label{thm:cs-lower}
	For any integer $C > 0$, any cut-preserving sketch for $k$-terminal undirected graphs, where the total capacity of edges incident on the terminals is equal to $C$, requires $\Omega(C/\log{C})$ bits.
\end{theorem}

To prove this lower bound, we show how to encode a binary vector of length $N := \Omega(C/\log{C})$ in an undirected graph $G$, so that given only a cut-preserving 
sketch for $G$, one can recover any entry of this vector. Standard information-theoretical arguments (similar to the lower bound for the \membership problem in Section~\ref{sec:matching}) 
then imply that size of the cut-preserving sketch has to be of size $\Omega(N) = \Omega(C/\log{C})$. 
We emphasize that while in the proof we assume the cut-preserving sketch has to return the value of minimum cuts between any subsets $(A,B)$ of the terminals (even when they are not a bipartition), by Claim~\ref{clm:general-cs}, 
this is without loss of generality; hence, the lower bound holds also for cut-preserving sketches that only guarantee to preserve minimum cuts for bipartitions.

\subparagraph{Construction.} Let $k' = k-2$. For simplicity, assume $k'$ is even, and let $N = {{k'}\choose{{k'\over2}}}$.
For any $N$-dimensional binary vector $\bv \in \set{0,1}^{N}$, we define a graph $G_{\bv}(V,E)$ as follows: 

\ssection{Vertices:} The set of vertices of $G_{\bv}$ is $V = \set{s,t} \union\set{q_1,\ldots, q_{k'}} \union \set{u_1,\ldots,u_{k'}} \union \set{v_1, \ldots , v_{N}}$ and the $k$ terminals are $T=\set{s,q_1,\ldots, q_{k'},t}$.

\ssection{Edges:} Let $\SS = \set{S_1,\ldots, S_N}$ be a collection
of all $({k'/2})$-size subsets of $\set{q_1,\ldots,q_{k'}}$. The set of edges are defined as:
\begin{itemize}
\item For any $i \in [k']$, there is an edge $(q_i,u_i)$ with capacity $N$.
\item For any $i \in [k']$, there is an edge $(s,u_i)$ with capacity $N$.
\item For any $j \in [N]$, there is an edge $(v_j,t)$ with capacity
  $1$.
\item \label{line:cal-parallel} A vertex $u_j$ is connected to a vertex $v_i$ with an edge of capacity $1$ iff $\bv_{i} = 1$ or $q_{j} \notin S_{i}$. Additionally, if $u_{j}$ is connected
  to $v_{i}$, there are two more edges $f_{1} = (s,v_{i})$ and $f_{2} = (u_{j},t)$ each with capacity $1$.
\item There is an edge $(s,t)$ with capacity $kN - m$, where $m$ is
  the number of edges between $\set{u_1,\ldots,u_{k'}}$ and
  $\set{v_1,\ldots,v_N}$.
\end{itemize}

To recover the vector $\bv$ from a cut-preserving sketch of $G_\bv$, we will consider the terminal cuts $(A,B)$
where $A = \set{s} \union S_{i}$ for some $S_i \in \SS$ and $B = \set{t}$. We further denote the terminal cut $(A,B)$ corresponding to picking $S_i \in \SS$ in the part $A$
by $\tc(S_i)$. We define the \emph{output profile} of a graph $G_\bv \in \FG$ to be an $N$-dimensional vector $\op(G_\bv)$ where
the $i$-th entry of $\op(G_\bv)$ is equal to the value of the terminal cut $\tc(S_i)$. We show that there is a one-to-one correspondence between the vector $\bv$ and $\op(G_\bv)$. 

\begin{lemma}\label{lem:flow-op}
  Let $\vec{1}$ be the $N$-dimensional vector of all ones. There exists a value $c$ independent of $\bv$ such that $\op(G_\bv) = \bv + c \cdot \vec{1}$.
\end{lemma}
\begin{proof}
  Fix an $i \in [N]$ and consider $\tc(S_i)$. We argue that the maximum flow value from $\set{s} \union{S_i}$ to $\set{t}$ is $(k+1)N-1 + \bv_i$; the lemma then follows
  from max-flow min-cut duality and the choice of $c = (k+1)N-1$.

  In $G_\bv$, we can first send a flow of size $kN$ from $s$ to $t$ by
  sending one unit of flow along every $s \rightarrow v_{l} \rightarrow u_{j} \rightarrow t $ path for any edge
  of the form $(u_{j},v_{l})$ ($m$ units of flow in total) and $kN - m$ units of flow over the $(s,t)$ edge.
  After this process, the residual graph of $G_\bv$ becomes a directed graph where any edge of the form $(u_{j},v_{l})$ is directed from $u_j$ to $v_l$.
  Now consider any vertex $v_p$ where $p \neq i$. There exists at least one terminal $q_j \in S_i$ (in fact, in $S_i \setminus S_p$), such that there is an
  edge between $u_j$ and $v_p$ in $G_\bv$. Since in the residual graph of $G_\bv$, this edge is directed from $u_j$ to $v_p$, we can send one unit of flow over this edge
  also through the path $q_j \rightarrow u_j \rightarrow v_p \rightarrow t$. Hence, in $G_\bv$, we can always send $kN + N-1 = (k+1)N - 1$ units of flow from $\set{s} \union S_i$ to $\set{t}$.

  First suppose the $i$-th entry of $\bv$ is equal to $1$; then there is an edges from $u_j$ to $v_i$ for any $q_j \in S_i$. In particular, we can send one extra unit of flow over one of these
  edges to $t$, hence having a flow of size $(k+1)N$ entering $t$. Since the total capacity of the edges incident on $t$ is $(k+1)N$, this ensures that the max-flow is also $(k+1)N$.
  		
  Now suppose the $i$-th entry of $\bv$ is equal to $0$. For the vertex $v_i$, by construction, there is no
  edge from any $u_{j}$ to $v_{i}$, where $q_j \in S_i$. Hence in the residual graph of $G_\bv$, there is no path from $\set{s} \union S_i$ to $\set{t}$, meaning that
  the maximum flow in this case is $(k+1)N - 1$. This completes the proof.
\end{proof}

\begin{proof}[Proof of Theorem~\ref{thm:cs-lower}]
Lemma~\ref{lem:flow-op} ensures that for any graph $G_\bv$, there is a one-to-one correspondence between the value of $i$-th entry in $\bv$ and $i$-th entry in $\op(G_{\bv})$. 
Assuming that the cut-preserving sketch is able to answer each terminal cut (deterministically or even with a sufficiently small constant probability of error), we can recover
$i$-th bit of $\bv_i$, from the $i$-th index in $\op(G_\bv)$ with a constant probability. Standard information-theoretical arguments imply that the size of 
the cut-preserving has to be $\Omega(N)$. Moreover, since $N = 2^{\Omega(k)} = \Omega(C/k) = \Omega(C/\log{C})$ in this construction, we obtain the final bound of $\Omega(C/\log{C})$ bits on the sketch size.
\end{proof}

We should point out for the case of randomized cut-preserving sketches that are only guaranteed to have a constant probability of failure over bipartition cuts (and not necessarily terminal cuts), we first need to reduce the probability 
of error to $2^{-k}$ before performing the described construction (and applying Claim~\ref{clm:general-cs}) which results in a lower bound of $\Omega(C/\log^{2}{C})$.

We further point out that, as a corollary of Theorem~\ref{thm:cs-lower}, we also obtain a simple proof for a lower bound of $2^{\Omega(k)}$ on size of 
cut-preserving sketches (see~\cite{krauthgamer2013mimicking,khan2014mimicking}).

\newcommand{\eijplus}{\ensuremath{\eplus_{(i,j)}}\xspace}
\newcommand{\eijminus}{\ensuremath{\eminus_{(i,j)}}\xspace}
\newcommand{\GQall}{\ensuremath{G^{\Qall}}\xspace}
\newcommand{\ehatplus}{\ensuremath{\hat{e}^{+}}\xspace}
\newcommand{\ehatminus}{\ensuremath{\hat{e}^{-}}\xspace}

\section{The $s$-$t$ Edge-Connectivity Problem}\label{sec:st-conn}
In this section, we study dynamic sketching for the $s$-$t$ edge-connectivity problem. As it turns out, any cut-preserving sketch can be directly adapted to a dynamic sketching scheme for the $s$-$t$ edge-connectivity problem
as follows. Given a graph $G$ with a set $T$ of $k$ terminals and two designated vertices $s$ and $t$, create a cut-preserving sketch for $G$ with terminals $T \union \set{s,t}$. Note that given a query $Q$ (i.e., a set of edges 
among $T$), the $s$-$t$ minimum cut (which is equal to the $s$-$t$ edge-connectivity) will partition $T \union \set{s,t}$ into two sets $T_{s}$ and $T_{t}$, where $T_{s}$ contains $s$ and $T_{t}$ contains $t$. Hence, the minimum 
cut from $T_{s}$ to $T_{t}$ is equal to the minimum cut from $s$ to $t$. The cut-preserving sketch can answer the minimum cut from $T_{s}$ to $T_{t}$ in the original graph, and the additional cut value caused by the query is simply 
the total number of the edges from $T_{s}$ to $T_{t}$. Therefore, if we enumerate all possible partitions of the terminals that separate $s$ and $t$, and compute the 	minimum cut for each partition as above, the smallest minimum 
cut among those partitions is equal to the minimum cut from $s$ to $t$.

Nevertheless, by our lower bound on the size of cut-preserving sketches (Theorem~\ref{thm:cs-lower}), a dynamic sketching scheme constructed as above, will have a linear dependency on the total degree of the vertices
 in $T \union \set{s,t}$, which could be as large as the number of vertices in the graph. To resolve this issue, we propose a scheme which directly uses our dynamic sketching scheme for the maximum matching  problem and achieve a sketch of size $O(k^4)$.

\subsection{A Dynamic Sketching Scheme with Sketch Size $O(k^4)$}\label{sec:st-conn-matching}
We now propose a dynamic sketching scheme by creating an instance of the maximum matching problem, and show that any query for the $s$-$t$ edge-connectivity problem can be answered using a query in the maximum matching problem. The reduction is in the same spirit as the one we used for cut-preserving sketches. But note that the main differences is that for cut-preserving sketches, the set of edges in the original graph never change, unlike the case for dynamic sketching.

 \begin{theorem}
   For any $\delta > 0$, there exists a randomized $k$-dynamic sketching scheme for the $s$-$t$ edge-connectivity problem with a sketch of size $O(k^4 \log(1/\delta))$, which answers any query correctly with probability at least $1 - \delta$.
\end{theorem}

Given a digraph $G$ with two designated vertices $s$ and $t$, along with a set of $T$ terminals, recall that, for the query $\Qall$ where an edge is inserted between each (ordered) pair of terminals, \GQall denotes the graph after 
applying the query $\Qall$ to $G$. Assume $s$ does not have any incoming edge and $t$ does not have any outgoing edge, since removing them will not affect the $s$-$t$ edge-connectivity. We further assume that $s$ are $t$  are 
not terminals. This is without loss of generality since we can create two vertices $s'$ and $t'$ that are not terminals, while adding $d^+(s)$ (resp. $d^{-}(t) $) new vertices and for each of these vertices $v$, adding an edge from $s'$ 
to $v$ and $v$ to $s$ (resp. $t$ to $v$ and $v$ to $t'$). In this new graph, the $s'$-$t'$ edge connectivity is equal to the $s$-$t$ edge-connectivity in $G$ and $s',t'$ are not terminals. Consider the following dynamic sketching 
scheme.

\textbox{A dynamic sketching scheme for the $s$-$t$ edge-connectivity problem}{
\begin{enumerate}
  \item[] \textbf{Input:} A graph $G$ with $m$ edges, two designated vertices $s$ and $t$, and a set $T$ of $k$ terminals.
  \item[$\bullet$] \textbf{Compression:} Construct a  bipartite graph $G'(L,R,E')$ with a set $T'$ of terminals  as follows and create a dynamic sketch for the maximum matching problem for $G'$ and $T'$.
    \begin{enumerate}
        \item For each edge $e$ in $\GQall$, if $e$ starts with $s$, create a vertex \eplus (in $R$), if $e$ ends with $t$ create a vertex \eminus (in $L$), otherwise, create two vertices \eplus (in $R$) and \eminus (in $L$).
        \item For each edge $e$ between two terminals in $G$, create two vertices \ehatplus and \ehatminus; \ehatminus and \ehatplus, along with  \eminus and \eplus,  are terminals of $G'$.
        \item For each edge $e$ in $G$ where both \eminus and \eplus exist, there is an edge between \eminus and \eplus.
        \item For any two edges $e_1$ and $e_2$ in $\GQall$ where the tail of $e_1$ is the head of $e_2$, there is an edge between $\eplus_1$ and $\eminus_2$.
    \end{enumerate}
    \item[$\bullet$] \textbf{Extraction:} Given any query $Q$ of $G$, let $Q'$ be the query of $G'$ where
    \begin{enumerate}
        \item For each edge $e$ in $Q$, insert an edge between \eminus and \eplus.
        \item For each edge $e$ in $\Qall \setminus Q$, insert an edge between \eminus and \ehatplus, and between \eplus and \ehatminus.
	\item Let the maximum matching size of $G'^{Q'}$ be $\nu$. Output $\nu - (m + 2k^2 - \card{Q})$.

    \end{enumerate}
\end{enumerate}
}

The total number of terminals in $G'$ is $4k^2$. Hence by Theorem~\ref{thm:matching}, the sketch size is  $O(k^4 \log(1/\delta))$. In the rest of this section, we prove the correctness.

\begin{proof}
    For any query $Q$ of $G$, and the query $Q'$ of $G'$ defined in the dynamic sketching scheme, we prove that $\nu - m - 2k^2 + \card{Q} = l$, where $\nu = \nu(G'^{Q'}$, i.e, the maximum matching size of $G'^{Q'}$, and $l$ is 
    the $s$-$t$ edge-connectivity of $G^Q$. Let $M$ be the matching in $G'^{Q'}$ where for each edge $e$ in $G^{Q}$, \eminus is matched with \eplus, and for each edge $e$ in $\Qall \setminus Q$, \eminus is matched with \ehatplus, and \eplus is matched with \ehatminus. The size of $M$ is equal to $m$ (the number of edges in $G$) plus $\card{Q}$ (the number of edges in $Q$) plus $2(k^2 - \card{Q})$ (the edges in $\Qall \setminus Q$), which is equal to $m + 2k^2 - \card{Q}$.

    We first argue that if the $s$-$t$ connectivity in $G^Q$ is $l$, then $\nu(G'^{Q'})\geq \card{M} + l$. This can be proved using the same argument as is used for cut-preserving sketches. Namely, there are $l$ edge disjoint-path paths from $s$ to $t$, which can be converted into $l$ vertex-disjoint augmenting paths of the matching $M$. We omit the details.

    We now show that if the size of a maximum matching $\mstar$ of $G'^{Q'}$ is $\nu$, then $l \ge \nu - \card{M}$. Then, combined with the previous case that $\nu \ge \card{M} + l$, we get $\nu = \card{M} + l$.  Denote $\nu = \card{M} + l'$ for some $l' \ge 0$; we will show that $l \ge l'$. The symmetric difference between $M$ and $\mstar$ forms a graph with $l'$ augmenting paths of the matching $M$.  We will show that these $l'$ augmenting paths can be converted into $l'$ edge-disjoint paths from $s$ to $t$ in $G^Q$. Each augmenting path must start and end with a vertex that is not matched in $M$, where the vertices that are not matched in $M$ can be partitioned into three types: $(i)$ $\eplus$ for outgoing edges of $s$, $(ii)$ \eminus for incoming edges of $t$, and $(iii)$ \ehatplus, \ehatminus for edges in $Q$. Since the Type~$(iii)$ vertices do not have any edge incident on them, they cannot be endpoints of the augmenting paths. Hence, the endpoints of the augmenting paths must be Type~$(i) $ or Type~$(ii)$. Since all Type~$(i)$ vertices are in $R$ and all Type~$(ii)$ vertices are in $L$, each augmenting paths must have one Type~$(i)$ endpoint and one Type~$(ii)$ endpoint (i.e., they start with $s$ and end with $t$).  Furthermore, since the degree of each \ehatplus, \ehatminus vertex is at most $1$, and they are not endpoints of the augmenting paths, none of them can appear in any augmenting path. Therefore, each augmenting path is of form $(\eplus_1, \eminus_2, \eplus_2, \eminus_3, \ldots, \eminus_j)$ where $e_1$ is an outgoing edge of $s$ and $e_j$ is an incoming edge of $t$. Therefore, the $l'$ augmenting paths form $l'$ edge-disjoint paths from $s$ to $t$, and $l \ge l'$.
\end{proof}

\subsection{Cut-Preserving Sketches and the $s$-$t$ Maximum Flow Problem}\label{sec:max-flow}
We have already established that any cut-preserving sketch can be adapted to a $k$-dynamic sketching scheme for the $s$-$t$ edge connectivity problem, but the size of the sketch will be $\tilde{\Omega}(C)$ which is not polynomial in $k$. However, directly constructing a scheme for edge-connectivity using our scheme for matchings achieves a sketch of size $\poly(k)$. We further explore this ``inefficiency'' of using cut-preserving sketches. As it turns out, there is an equivalence between cut-preserving sketches and dynamic sketching schemes for the $s$-$t$ \emph{maximum flow} problem.

The $s$-$t$ maximum flow problem is the capacitated version of the $s$-$t$ edge connectivity problem. We consider the case when the capacities of the edges between the terminals are fixed and given to the compression algorithm, and a query only reveals  which subset of these edges to be inserted to the graph. In this case, one can simply make the graph uncapacitated by replacing any edge $e = (u,v)$, whose capacity is $c_e$, with $c_e$ uncapacitated two-hop paths from $u$ to $v$. It is then straightforward to obtain a sketching scheme for the $s$-$t$ maximum flow problem using the sketch for the $s$-$t$ edge connectivity problem. Note that the size of the resulting sketch depends on $C$ and hence is not compact according to our definition. We establish the following theorem.
\begin{theorem}\label{thm:equivalence}
  Any cut-preserving sketch can be adapted to a dynamic sketching scheme for the $s$-$t$ maximum flow problem while increasing the number of terminals by at most $2$, and vice versa.
\end{theorem}
\begin{proof}
First of all, the argument we used for showing that cut-preserving sketches give dynamic sketching schemes for $s$-$t$ edge-connectivity (while increase the number of terminals by $2$) directly  works for max-flow. In the following, we show how to use a scheme for max-flow to create a  cut-preserving sketch.

We define the \emph{$\stext$} of a graph $G$ with a set $T$ of terminals to be a graph $H$ obtained by adding two
new vertices $s$ and $t$ to $G$ and connecting $s$ to each terminal $q \in T$ with an edge of capacity $c^{+}(q)$ (i.e, the total capacity of the outgoing edges of $q$) and connecting $q$ to $t$ with an edge of capacity $c^{-}(q)$ (i.e., the total capacity of the incoming edges of $q$). We establish the following simple Lemma.
\begin{lemma}\label{prop:cs-sketch}
  For any graph $G$ with a set $T$ of $k$ terminals; Let $H$ be the $\stext$ of $G$; then, any $k$-dynamic sketching scheme for the $s$-$t$ maximum flow problem for the graph $H$ with terminals $T \union
  \set{s,t}$ can be used as a cut-preserving sketch for $G$ with terminals $T$.
\end{lemma}
\begin{proof}
  Suppose $\Gamma$ is the $k$-dynamic sketch for max-flow for the graph $H$ and the terminals $T \union
  \set{s,t}$.  We show how to use $\Gamma$ as a cut-preserving sketch for the original graph $G$ and the terminals $T$.  For any $A,B \subseteq T$, let $Q$ be the query in which $s$ is only connected to the terminals $q_{i} \in A$ (with edges of capacity $c^{+}(q_i)$) and only the terminals $q_{j} \in B$ are connected to $t$ (with edges of capacity $c^{-}(q_j)$). By the min-cut max-flow duality, the $s$-$t$ maximum flow in $H^{Q}$ is equal to the minimum cut from $s$ to $t$, which in turn is equal to the minimum cut from $A$ to $B$ in $G$.
\end{proof}
\end{proof}

We point out that Theorem~\ref{thm:equivalence} combined with Theorem~\ref{thm:cs-lower}, proves a similar $2^{\Omega(k)}$ lower bound on size of dynamic sketches for the $s$-$t$ maximum flow problem. In 
other words, this problem does not admit a compact dynamic sketch.

\section{The Minimum Spanning Tree Problem}\label{sec:mst}
In this section, we provide a $k$-dynamic sketching scheme for the 
minimum spanning tree problem, proving the following theorem. 
At the end of this section, we discuss how to extend the sketch to work for the \emph{minimum spanning forest} problem when $G$ is not connected.

\begin{theorem}\label{thm:mst}
  There exists a deterministic $k$-dynamic sketching scheme for the
  minimum spanning tree problem with a sketch of size $O(k)$.
\end{theorem}

Note that we can, without loss of generality, assume that all edges have distinct weights; otherwise, we can use a deterministic tie-breaking rule when
comparing edges with the same weight. Consequently, for any query $Q$, the minimum spanning tree of $G^{Q}$ is unique. 

\subparagraph{Notation.} For any two vertices $s$ and $t$
in a spanning tree $H$ of a graph, we denote by $P^H(s,t)$ the unique path between $s$ and
$t$. When the underlying spanning tree is clear from the context, we will drop the
superscript $H$. We say that $P(s,t)$ is a \emph{clean} path,
if it only consists of non-terminals  with degree equal to 
two (excluding $s$ and $t$).  A clean path which is not contained in
any other clean paths is considered \emph{maximal}.

Our scheme for the minimum spanning tree problem is as follows. 

\textbox{A dynamic sketching scheme for the minimum spanning tree problem}{
\begin{enumerate}
	\item[] \textbf{Input:} An undirected weighted graph $G(V,E)$, with a set $T$ of $k$ terminals. 
	\item[$\bullet$] \textbf{Compression:} 
	\begin{enumerate}
 		\item \label{line:first-line} Let $H$ be the minimum spanning tree of $G$. 
 		\item \label{line:prune} Perform the following two steps to prune
   		$H$.
 		\begin{enumerate}
			  \item \label{line:prune-1} While there is a non-terminal leaf $v$ in  $H$, remove $v$ along with the edges incident on $v$.
			  \item \label{line:prune-2} While there is a maximal clean path
			    $P(s,t)$ in $H$ with at least $2$ edges, replace the whole path
			    $P(s,t)$ with an edge $e = (s,t)$ whose weight is defined as the
			    maximum edge weight in $P(s,t)$.
		 \end{enumerate}
 	\item Let $H'$ be the resulting tree. Store $H'$ together with the
   	difference between the total weights of $H$ and $H'$, denoted by
   	$w^*$ as the sketch $\Gamma$.
	\end{enumerate}
	\item[$\bullet$] \textbf{Extraction:}
	\begin{enumerate}
 		\item Given a query $Q$, let $H'^{Q}$ be the graph consisting of
 		  $H'$ and the edges in $Q$.  Note that $H'$ contains all $k$ terminals.
 		 \item Return $w^*$ added to the weight of the minimum spanning tree of
 		  $H'^{Q}$ .
	\end{enumerate}
\end{enumerate}
}

\begin{lemma}\label{lem:mst-size}
  Given a graph with a set of $k$ terminals, the sketch
  created by the compression algorithm requires storage of $O(k)$.
\end{lemma}
\begin{proof}
  It suffices to show that after compression, $H'$ has at most $O(k)$
  vertices. By construction, the tree $H'$ has at most $k$ leaf nodes
  (only terminals can be leaves) and at most $k$ vertices whose degree
  is equal to two (removing maximal clean paths ensures only a
  terminal can have degree equal to two).  This implies that there can
  be only $O(k)$ vertices in $H'$.
\end{proof}

To analyze the correctness of the sketch, we assume the following
algorithm is used for finding the minimum spanning tree in
the extraction algorithm. \mstalg\ takes a tree $H$ and a set of edges (a
query) $Q$ as input, and output the minimum spanning tree of the
graph consisting of both $H$ and $Q$. Note that in general, any
algorithm for finding minimum spanning tree can be used by the extraction algorithm
and we use \mstalg\ only for the sake of argument.

\textbox{{\mstalg$(H,Q)$}.}
{
\begin{enumerate}
    \item Let $H_0 = H$.
    \item {For} {$i=1$ to $\card{Q}$} {do}
     \begin{enumerate}[(a)]
     \item Let $e_i$ be the $i$'th edge in $Q$
      \item \label{line:remove} Add $e_i$ to $H_{i-1}$ and from the unique cycle created, remove the edge with the maximum weight. 
      \item Let this new tree be $H_i$.
     \end{enumerate}
    \item {Return} $H_{\card{Q}}$.
    \end{enumerate}
  }

It is an easy exercise to show that $\mstalg$ indeed computes a
minimum spanning tree of the graph consisting of both $H$ and $Q$.
Assume $H$ is the minimum spanning tree of the initial graph $G$; in
the following, we will argue that running $\mstalg$ on both $H$ and
$H'$ with the same query $Q$ will result in edges with the same
weight being removed in line (\ref{line:remove}). Therefore, the
difference between the weights of the two minimum spanning trees
returned by $\mstalg$ is still $w^*$.

We further define the following notation. Given a query $Q$, if we
run $\mstalg$ with input $H$ (the minimum spanning tree of $G$) and
$H'$ (the tree stored in the sketch) with query $Q$ in parallel,
denote by $H_{i}$ (resp. $H'_{i}$) the tree created in $\mstalg$ with
input $H$ (resp.  $H'$) after iterating the $i$-th edge of
$Q$. Furthermore, let $H_0 = H$ and $H'_0 = H'$. For any pair of $H_i$
and $H_i'$, we say an edge $e$ in $H_i'$ is a \emph{summary} of the
path $P^{H_i}(s,t)$ in $H_i$ iff the weight of $e$ in $H_i'$ is equal to the
maximum edge weight in path $P^{H_i}(s,t)$ in $H_i$.

\begin{lemma}\label{lem:mst-property}
Let $\ell = \card{Q}$; for any $i \in [\ell]$, and every edge $e = (u,v) \in H'_i$,
$e$ is a summary of the path $P^{H_i}(u,v)$ in
$H_i$. Moreover, any two edges in $H'_i$ are summaries 
of two \emph{edge disjoint paths} in $H_i$.
\end{lemma}
\begin{proof}
We prove the lemma by induction on $i$. 
In the compression algorithm, if a
maximal clean path is replaced by an edge $e$ in line
(2b), $e$ will never be  replaced.  Thus any edge
$e$ in $H'_0 ( = H')$ is either an edge in $H$ (the minimum spanning
tree of $G$), or a summary of the maximal clean path replaced by
$e$. Since the maximal clean paths are edge disjoint by definition, the
lemma holds for $i=0$.

Suppose the properties hold for $i=j$, and consider $i = j + 1$.  Let
$e_{j+1} = (u, v)$. Since $H'_j$ is connected, there is a path from
$u$ to $v$ in $H'_{j}$, $P^{H'_j}(u,v)$. By induction, the edges in
$P^{H'_j}(u,v)$ are summaries of edge disjoint paths in $H_j$. These
edge disjoint paths form the simple path $P^{H_j}(u,v)$
together. Hence, for the two cycles respectively formed by adding
$e_{j+1}$ to $H'_j$ and $H_j$, the largest weights of the edges are
the same. Since edge weights are unique, either $e_{j+1}$ is missing
in both $H'_{j+1}$ and $H_{j+1}$, or an edge $(u', v')$ is missing in
$H'_{j+1}$ and an edge with the same weight in $P^{H_j}(u', v')$ is
missing in $H_{j+1}$. For the later case, the newly added edge
$e_{j+1}$ trivially satisfies the two properties, and any edge
remained in $H'_{j}$ is a summary of the same path in both $H'_j$ and
$H'_{j+1}$. 
\end{proof}

\begin{lemma}\label{lem:mst-correct}
For any query $Q$, the extraction algorithm returns the weight of the minimum spanning tree of $G^{Q}$.
\end{lemma}

\begin{proof}
Assume we run $\mstalg$ on $H$ (the minimum spanning tree of $G$) and
$H'$ with query $Q$ in parallel. Recall that $H_i$ and $H'_i$ are the
resulting trees after iterating the $i$-th edge in $Q$ for $H$ and $H'$,
respectively.

When we add the $i$-th edge $e_i = (u,v)$ in $Q$ to $H_{i-1}$
(resp. $H'_{i-1}$), the removed edge $e$ (resp. $e'$) is either $e_i$
itself or the edge with the maximum weight on the path $P^{H_{i-1}}(u,
v)$ (resp. $P^{H'_{i-1}}(u,v)$).  By Lemma~\ref{lem:mst-property}, the
maximum weight of the edges in the path $P^{H'_{i-1}}(u, v)$ is equal
to the maximum weight of the edges in the path $P^{H_{i-1}}(u,v)$.  It
implies that the weight of $e$ is equal to the weight of $e'$.
Consequently, the total weight of all the removed edges is the same
for $H$ and $H'$. By adding $w^*$, which is equal to the difference
between the weights of the initial trees $H$ and $H'$, to the total weight
of $H'_{\card{Q}}$, we ensure that the returned weight is equal to the
weight of the minimum spanning tree of $G^Q$.  
\end{proof}

Theorem~\ref{thm:mst} is an immediate consequence of
Lemma~\ref{lem:mst-size} and Lemma~\ref{lem:mst-correct}.

For the case when $G$ is not connected, the only change we need to
apply is to initialize $H$ (line (1) of the
compression algorithm) with the minimum spanning forest of
$G$. Consequently, the $k$ terminals are distributed among different
trees in the forest, and by Lemma~\ref{lem:mst-size}, the number of
vertices in each tree is linear to the number of terminals inside that
tree. Hence, the size of the sketch is $O(k)$.

We still analyze the extraction algorithm using $\mstalg$. As long as a cycle is created when 
an edge from the query is added, each tree can be treated individually and the same argument
ensures correctness. If an edge $e_i$ in $Q$ connects two trees, adding
$e_i$ to either $H_{i-1}$ or $H'_{i-1}$ will not create a cycle,
thus no edge will be removed in line (\ref{line:remove}) of $\mstalg$.
Lemma~\ref{lem:mst-property} still holds and the argument in
Lemma~\ref{lem:mst-correct} can be directly applied to prove
the correctness.

\section{Limitations of Dynamic Sketching}\label{sec:limit}

In this section, we provide two simple lower bounds in the dynamic sketching model. In particular, we first show that some NP-hard problems like 
maximum clique or minimum vertex cover does not admit compact dynamic sketches and then give a similar lower bound for a problem which is in P. Together with our lower bound
for the maximum flow problem (see the discussion after Theorem~\ref{thm:equivalence}), these results demonstrates some of the limitations of applying the dynamic sketching framework to various problems.

\subsection{Lower Bounds for the Clique and Vertex Cover Problems}\label{sec:np}
We present a simple lower bound result on the size of any $k$-dynamic sketch for the maximum clique problem. Specifically, we show that any $k$-dynamic sketching scheme for outputting 
whether the input graph has a clique of size $\floor{k/2}+1$ or not requires sketches of size $2^{\Omega(k)}$.
A similar construction can also be used to establish identical lower bounds for the maximum independent set problem and the minimum vertex cover
problem, using the close connection (in terms of \emph{exact} answers) between these problems.
We emphasize that the lower bound is information-theoretic; it holds even if the compression and extraction algorithms are computationally unbounded.

The clique problem is to find a complete sub-graph (a clique) with the size $b = \floor{k/2} + 1$ in the
given input graph. 
We use a reduction from the \membership problem defined in Section~\ref{sec:matching}. 

\subparagraph{Reduction.} Let $N := {{k}\choose{\floor{k/2}}}$; for any set $S \subseteq [N]$, we define a graph $G_S$ with $N$ vertices $\set{v_1, v_2, \ldots , v_N}$ plus $k$ terminals $\set{q_1,\ldots,q_k}$. 
Let $\SS$ be the collection of all subsets of terminals $\set{q_1,\ldots,q_k}$ with size exactly $\floor{k/2}$ and fix an arbitrary bijection $\sigma: [N] \mapsto \SS$.   
For any $e \in [N]$, if $e \in S$, connect $v_e$ to all $\floor{k/2}$ terminals $q_i \in \sigma(e)$.

Alice creates the graph $G_S$ and compress it using the $k$-dynamic sketching scheme for the clique problem. She
will then send the dynamic sketch to Bob. Bob queries the sketch with
the query $Q$ whereby the $\floor{k/2}$ terminals corresponding to $\sigma(e^*)$ form a clique. 

We now argue the correctness of this reduction. 

\begin{proof} If $e^* \in S$, adding the vertex $v_{e^*}$ in $G_S$ to the clique in the query $Q$
 forms a clique on $\floor{k/2} + 1$ vertices, hence $G_S$ has a clique of size $\floor{k/2} + 1$ in this case. 
 
 If $e^* \notin S$, the maximum clique has $\floor{k/2}$ vertices as there is no edges between $v_i$'s vertices and hence any clique can contain only one of them.  

The lower bound of $\Omega(N) = 2^{\Omega(k)}$ on the sketch size now follows from this reduction and $\Omega(N)$ lower bound of the \membership problem. 
\end{proof}
\subsection{A Lower Bound for Boolean Functions in P}\label{sec:boolean}
We show that there are Boolean functions in P for which any
$k$-dynamic sketch requires $\min\set{2^k,n}$ space.  Consider a
Boolean function $f$ where the $n$-bit static part represents the
truth table of a $k$-ary Boolean function $f'$ (hence $n = 2^k$), and
the $k$-bit dynamic part represents the input to $f'$. Clearly, $f$
can be computed in poly-time. Moreover, any compressed representation
for $f$ can be used to recover $f'$ by trying all possible $2^k$
updates on the dynamic part. Since the number of different $k$-ary
Boolean functions is $2^{2^k}$ , any compressed representation for $f$
requires to store at least $2^k = n$ bits. On the other hand, for any
Boolean function $f$, there is always a trivial sketch of size
$\min\set{2^k,n}$ that allows one to compute $f$ for any possible
update to a specified set of $k$ elements -- one can either store
answers for all possible $2^k$ updates or let the entire dataset be
the sketch.

\section{Conclusions}\label{sec:conclusion}
In this paper we have introduced \emph{dynamic sketching}, a new
approach for compressing data sets separated into static and dynamic
parts. We studied dynamic sketching for graph problems where the
dynamic part consists of $k$ vertices and the edges between them
may get modified in an arbitrary manner (a query).  We showed that
the maximum matching problem admits a sketch of size $O(k^2)$ and the space bound is tight. Moreover,
this sketch can be used to obtain cut-preserving sketches of size $O(kC^2)$, and dynamic sketches for $s$-$t$ 
edge-connectivity of size $O(k^4)$.

There are problems (even in P) for which any dynamic sketch requires
$2^{\Omega(k)}$ space. An interesting direction for future work is to
identify broad classes of problems that admit compact dynamic sketches, i.e, sketches of size
$\poly(k)$.

Some data compression schemes (most notably, cut sparsifiers and
kernelization results) generate as compressed representation an
instance of the original problem, while the sketches we introduced do not fall into this category.  
A natural question is to understand if there exist
polynomial-size ``sparsifier-like'' compressed representations for
matchings and $s$-$t$ edge connectivity in the dynamic sketching
model.

Finally, while our work narrows the gap between upper and lower bounds
on the size of a cut-preserving sketches, it remains an intriguing open question to get an
asymptotically tight bound on the size of cut-preserving sketches.

\subsection*{Acknowledgments} We are grateful to Chandra Chekuri and Michael Saks for
helpful discussions.

\clearpage
\bibliographystyle{plain}
\bibliography{ref}

\end{document}